\pdfoutput=1
\documentclass[11pt]{article}

\usepackage[margin=1.1in]{geometry}
\usepackage{amsmath,amssymb,amsthm}
\usepackage{booktabs}
\usepackage{url}
\usepackage{lmodern}
\usepackage{microtype}
\usepackage[hidelinks]{hyperref}
\hypersetup{
  pdftitle={New lower bounds A(23,6,10) >= 2979 and A(24,6,10) >= 4214 for
    binary constant-weight codes},
  pdfauthor={Christian Lysenst{\o}en},
}

\newtheorem{theorem}{Theorem}
\newtheorem{proposition}[theorem]{Proposition}
\theoremstyle{remark}
\newtheorem{remark}[theorem]{Remark}

\newcommand{\A}{A}

\title{New lower bounds for binary constant-weight codes:\\
$A(23,6,10)\ge 2979$ and $A(24,6,10)\ge 4214$}

\author{Christian Lysenst{\o}en\\[4pt]
\small Inland Norway University of Applied Sciences\\[2pt]
\small\texttt{christian@lysenstoen.net}}

\date{July 2026}

\begin{document}
\maketitle

\begin{abstract}
Let $\A(n,d,w)$ denote the maximum size of a binary constant-weight code of
length $n$, minimum distance $d$, and weight $w$. We construct explicit codes
proving $\A(23,6,10)\ge 2979$ and $\A(24,6,10)\ge 4214$. These improve the
best surviving explicit codes of sizes 2969 and 4174 and surpass the
corresponding 1990 bounds 2970 and 4200 of Brouwer, Shearer, Sloane and Smith,
whose code listings were lost. We also obtain $\A(23,6,11)\ge 3539$ and
$\A(24,6,8)\ge 1855$. All four codes were independently verified by the
table's maintainer, whose local polishing subsequently improved them to 2981,
3550, 1857, and 4216, the values now on the live table. The constructions use
a coordinate decomposition in which one half is
fixed to a known code and the complementary half is selected from its full
cross-compatible pool using CHILS for maximum-weight independent set. For the
2969-word $\A(23,6,10)$ incumbent, exact computations with two solver families
prove insertion maximality and exclude every improving exchange deleting at
most three codewords. We also analyze codes invariant under prime-order
permutations: several cycle types are excluded exactly, the $5+1^{18}$ type
has upper bound 499, and reproducible heuristic saturation evidence is
reported for the remaining types, with $13+1^{10}$ left open. Code files, an
independent validator, model descriptions, and computational logs are
released.

\end{abstract}

\section{Introduction}\label{sec:intro}

A binary \emph{constant-weight code} with parameters $(n,d,w)$ is a set of
binary vectors of length $n$, each of Hamming weight exactly $w$, pairwise at
Hamming distance at least $d$; $\A(n,d,w)$ denotes the maximum size of such a
code. For even $d$, distance is equivalent to a bounded pairwise intersection:
two weight-$w$ words $u,v$ satisfy $d(u,v)\ge d$ if and only if
$|u\cap v|\le w-d/2$, where words are identified with their supports.
Brouwer's online table~\cite{brouwer-table} is the community reference for
lower and upper bounds on $\A(n,d,w)$.

The 1990 paper of Brouwer, Shearer, Sloane and Smith~\cite{bsss1990}
introduced many bounds by heuristic search. Nine of its $d=6$ entries carry
the label ``y'' (no known structure) and exceed 1500 codewords, so their
explicit listings were omitted from the paper; the underlying machine outputs
did not survive. On the modern table these nine cells were printed as
\emph{lost codes}. A public-archive search (Wayback captures of Sloane's AT\&T
table, Shearer's pages, Edel's repository, and Brouwer's own tree) finds no
surviving explicit file for any of them.

Two facts frame the difficulty. First, the \emph{base rate}: verified against
the live table, none of the nine lost bounds had been re-attained in 35 years,
and Brouwer's page documents his own reconstruction effort (``we construct new
codes with [BSSS] parameters or better, with 9 exceptions'' --- these nine are
the exceptions). The table's maintainer already tried, and did not succeed, on
exactly this list. Second, the \emph{possibility of error}: a lost bound can
be an overcount --- $\A(18,6,6)$ was printed as 144 in the same 1990 table and
stands at 133 today --- so a hard search plateau below a lost value is itself
informative, and a lost value can only be settled by re-attaining it or by an
impossibility proof beyond current technology.

Against that backdrop, this paper reports:
\begin{enumerate}
\item explicit codes surpassing two of the nine lost bounds,
  $\A(23,6,10)\ge 2979$ and $\A(24,6,10)\ge 4214$, and improving two further
  cells over their best surviving codes, $\A(23,6,11)\ge 3539$ and
  $\A(24,6,8)\ge 1855$ (Section~\ref{sec:bounds}) --- all four verified by
  the table's maintainer, and subsequently further improved by his polishing
  (see the note closing Section~\ref{sec:data});
\item exact local-exchange computations and measured saturation evidence that
  help explain the difficulty of improving the surviving 2969-word incumbent:
  exchange dryness to deletion depth 3 (Section~\ref{sec:exchange}) and a
  cycle-type analysis of prescribed automorphisms of prime order $\ge 5$
  (Section~\ref{sec:symmetry});
\item a faithful reimplementation of the 1990 improver described in
  \cite[\S XII]{bsss1990}, augmented with exact plateau computations
  (Section~\ref{sec:bounds}); and
\item a quantification of the reproduction distance of the lost results
  (Section~\ref{sec:distance}).
\end{enumerate}
Section~\ref{sec:verification} states the verification discipline;
Section~\ref{sec:data} lists the released artifacts.

Throughout, every numeric claim is either an exact computational result
(reproduced by two independent solver families where stated), a
\textsc{measured} saturation (a heuristic search cap reproduced across
independent runs), or a \textsc{sourced} historical fact; the text labels
which.

\section{Preliminaries and verification discipline}\label{sec:verification}

Throughout, the \emph{incumbent} is the surviving code
\texttt{a23.6.10.2969H} of 2969 words. For weight-10 words in 23 coordinates,
two words are compatible if and only if $|u\cap v|\le 7$.

\paragraph{Two independent code-validity paths.}
Every code is checked by two computationally independent verifiers: one
computes pairwise Hamming distance by XOR-popcount, the other computes
pairwise intersection counts. A code counts only if both agree on validity,
size, and distinctness. For the four bounds reported here we additionally
provide a third, dependency-free verifier
(\path{anc/independent_verify.py}, pure Python standard library,
sharing no code with the construction pipeline) that recomputes minimum
distance two further ways (string position-difference and integer XOR) and
requires both to agree. Thus a referee can re-check any code file without
using the construction software. Code validity is a finite exact
computation.

\paragraph{Two independent solver engines.}
Every exchange or symmetry statement labeled a theorem below is computed as
\textsc{infeasible} (or \textsc{optimal}) by two independent solver families
--- CP-SAT~\cite{ortools} and one of CBC~\cite{cbc} or SCIP~\cite{scip} ---
encoding the identical model.

\paragraph{Self-test on ground truth.}
Before any exchange verdict is trusted, the identical pipeline must pass a
two-sided self-test on a cell with known optimum, $\A(18,6,6)=133$:
(i)~hold-one-out --- with one word removed from the optimum, the solver must
find the $+1$ witness; (ii)~maximality --- on the full 133-word code it must
prove \textsc{infeasible}. Both directions passed on every run.

\paragraph{Record discipline.}
No result was called a record until it passed an offline gate (two independent
verifiers plus distinctness) against a same-day re-pull of the live table,
and, for the word ``record'' in the community sense, until the table's
maintainer accepted the submission. The four bounds here passed the gate
against the table as re-pulled on 20 July 2026, were submitted to the
maintainer the same day, and were verified and added to the live table in July
2026; the corresponding code files are now also hosted in the table's own code
archive.

\section{Local optimality of the incumbent}
\label{sec:exchange}

Write the incumbent as $C_0$ (2969 words). A \emph{candidate} for insertion is
any weight-10 word not in $C_0$; its \emph{blockers} are the incumbent words
it conflicts with (intersection $>7$).

\begin{theorem}[insertion maximality]\label{thm:maximal}
No candidate word has zero blockers; $C_0$ is insertion-maximal.
\end{theorem}

\begin{proof}[Computational verification]
Blocker census over all $\binom{23}{10}-2969=1{,}141{,}097$ candidate words
(the words of $C_0$ excluded); the zero-blocker count is exactly 0.
\end{proof}

\begin{theorem}[$C\le 2$ exchange dryness]\label{thm:c2}
There is no set $S$ of candidate words, each with at most two blockers,
pairwise compatible, with $|S| > \bigl|\bigcup_{s\in S}\mathrm{blockers}(s)\bigr|$;
that is, no positive exchange exists whose inserted words each conflict with
at most two incumbent words (blockers counted with respect to $C_0$),
regardless of the number of deletions.
\end{theorem}

\begin{proof}[Computational verification]
CP-SAT maximizing the gain
$|S|-\bigl|\bigcup_{s\in S}\mathrm{blockers}(s)\bigr|$ over the 1248
candidates with at most two blockers, subject to 5261 pairwise-conflict
constraints; \textsc{optimal} with objective value 0, 0.1\,s.
\end{proof}

\begin{theorem}[depth-3 hole-refill dryness]\label{thm:depth3}
There is no pair $(D,S)$ with $D\subseteq C_0$, $|D|\le 3$, and $S$ a set of
weight-10 words, such that $(C_0\setminus D)\cup S$ is a valid $(23,6,10)$
code with $|S|>|D|$.
\end{theorem}

\begin{proof}[Computational verification]
This is the literal 1990 move of cutting a hole in the code and refilling it,
solved exactly at deletion depth 3. CP-SAT feasibility model with binary
variables $x_s$ (insert candidate $s$) and $y_c$ (delete incumbent word $c$):
insert variables restricted to the 7751 candidates with at most 3 blockers
--- sufficient, since any inserted word in a solution with $|D|\le 3$ has at
most 3 blockers --- linking constraints $x_s\le y_c$ for every blocker $c$ of
$s$, deletion budget $\sum_c y_c\le 3$, pairwise-conflict cuts
$x_s+x_{s'}\le 1$ ($43{,}105$ pre-seeded for shared-blocker pairs), and the
cardinality cut $\sum_s x_s\ge\sum_c y_c+1$; \textsc{infeasible} in 4.1\,s,
reproduced three times. Independently re-proven by CBC branch-and-cut on the
identical model, \textsc{infeasible} in 784\,s.
\end{proof}

\paragraph{Soundness of the relaxation.}
Lazily omitted pairwise-conflict constraints only enlarge the feasible region;
\textsc{infeasible} on the relaxation is therefore a valid proof for the true
problem.

\paragraph{Open problem.}
The depth-4 analogue is unresolved: the model (30{,}247 candidates,
${\approx}888$k pre-seeded pairs) was not closed within the budgets tried
(CP-SAT, CBC, and SCIP each left it open), so we state it as an open problem
rather than a theorem.

Theorems~\ref{thm:maximal}--\ref{thm:depth3} help explain the long plateau
at 2969 in the language of the heuristic that produced the lost value: the
standard local move, exhaustively searched to depth 3, opens nothing.

\section{Symmetry exclusions: evidence against low-order automorphisms}
\label{sec:symmetry}

A standard construction method in this regime --- prescribed automorphisms in
the Kramer--Mesner tradition~\cite{kramer-mesner,bjlo2019} --- fixes a group
$G\le S_{23}$ and searches unions of whole $G$-orbits. We investigate codes
invariant under permutations of prime order at least 5. Several cycle types
are excluded exactly, one admits a direct analytic bound, and the remaining
types are studied by reproducible heuristic saturation experiments, with the
$13+1^{10}$ type left open.

Since conjugation is a coordinate relabeling that preserves code size and
validity, the $G$-invariant search space depends on a cyclic subgroup only
through its cycle type. On 23 points an element of prime order $p$ has a cycle
type consisting of $p$-cycles and fixed points. We give a verdict for each
such type.

Fix a permutation $g$ and let $\langle g\rangle$ act on the weight-10 words.
Call an orbit \emph{self-compatible} if its members are pairwise compatible.
A $g$-invariant code is a union of orbits, and an orbit containing a
conflicting pair can never be contained in a code, so every $g$-invariant
code is a union of self-compatible orbits. The \emph{orbit-conflict graph}
has the self-compatible orbits as vertices, weighted by orbit size, with an
edge whenever some word of one orbit conflicts with some word of the other;
$g$-invariant codes of size $s$ correspond exactly to independent sets of
total weight $s$. An exact bound on the maximum independent-set weight
therefore bounds every $g$-invariant code.

\begin{theorem}[$11+1^{12}$ exclusion]\label{thm:t11}
No $(23,6,10)$ code of size $\ge 2970$ is invariant under a permutation of
cycle type $11+1^{12}$.
\end{theorem}

\begin{proof}[Computational verification]
The 104{,}066 orbits reduce to 2640 self-compatible ones; a maximum-weight
independent set of weight $\ge 2970$ in the orbit-conflict graph is
\textsc{infeasible} --- CP-SAT, 19\,s; independently re-proven by SCIP
branch-and-cut in 2\,s; the orbit-graph construction is cross-validated by two
independent implementations.
\end{proof}

\begin{theorem}[$7+1^{16}$ exclusion]\label{thm:t7}
The same holds for cycle type $7+1^{16}$.
\end{theorem}

\begin{proof}[Computational verification]
The only self-compatible orbits are the 8568 fixed words, and
\textsc{infeasible} holds over the 3.0M-pair conflict graph (CP-SAT;
independently re-proven by SCIP in 21\,s).
\end{proof}

\begin{proposition}[inheritance]\label{prop:inherit}
If $G\le S_{23}$ has order divisible by 11, then by Cauchy's theorem $G$
contains an element of order 11, whose cycle type on 23 points is $11+11+1$
or $11+1^{12}$, and any $G$-invariant code is invariant under that element.
Likewise a transitive $G\le S_{23}$ has order divisible by 23
(orbit--stabilizer), hence contains a 23-cycle.
\end{proposition}

Type $11+1^{12}$ is excluded exactly by Theorem~\ref{thm:t11}. For type
$11+11+1$ the invariant space is conjugate to a fixed one whose best
saturation is ${\approx}2475$ (\textsc{measured}, dense-graph search), and
for the 23-cycle ${\approx}2369$ (\textsc{measured}) --- both well below
2970, though these are heuristic caps, not proofs. Together with
Theorem~\ref{thm:t11}, these measured caps provide evidence
against $M_{22}$ point stabilizers, $\mathrm{PSL}(2,11)$ acting $12+11$, all
Frobenius groups $11{:}k$, and transitive groups as sources of a code meeting
the target.

\paragraph{Prime-cycle-type sweep.}
For the remaining prime cycle types we report the orbit census and the best
search saturation (\textsc{measured}):

\begin{center}
\begin{tabular}{@{}llll@{}}
\toprule
cycle type & orbits & self-compatible & verdict (vs.\ target 2970) \\
\midrule
$23$ (transitive) & 49{,}742 & 49{,}170 & measured cap ${\approx}2369$ \\
$11+11+1$ & 104{,}006 & 103{,}059 & measured cap ${\approx}2475$ \\
$11+1^{12}$ & 104{,}066 & 2640 & \textbf{Theorem~\ref{thm:t11}: infeasible} \\
$7+1^{16}$ & 170{,}782 & 8568 & \textbf{Theorem~\ref{thm:t7}: infeasible} \\
$5+1^{18}$ & 270{,}674 & 52{,}326 & cap 499 (Remark~\ref{rem:fivecycle}) \\
$13+1^{10}$ & 88{,}006 & 26{,}720 & open, leaning infeasible$^{\dagger}$ \\
$17+1^{6}$, $19+1^{4}$ & 67k, 60k & 55k, 56k & measured caps ${\approx}1824$--$1836$ \\
multi-cycle 7- and 5-types & 163k--231k & 53k--226k & measured caps ${\approx}1001$--$1848$ \\
\bottomrule
\end{tabular}
\end{center}
\noindent {\footnotesize $^{\dagger}$greedy reaches 1457; a solve with only
part of the pairwise-conflict constraints enforced (a \emph{constraint
relaxation}) packs 2977 orbit words, but greedily deleting words until no
conflict remains (\emph{repair}) retains only ${\approx}871$.}

\begin{remark}\label{rem:fivecycle}
For type $5+1^{18}$ the census shows the self-compatible orbits are precisely
the $\binom{18}{10}+\binom{18}{5}=52{,}326$ fixed words. A fixed word's
support either avoids the 5-cycle (a weight-10 word on the 18 fixed points)
or contains it (the 5-cycle plus a weight-5 word on the fixed points). Two
avoiders are compatible exactly when their restrictions form an $(18,6,10)$
pair; two containers share the 5-cycle, so they are compatible exactly when
their weight-5 parts intersect in at most 2 points, i.e.\ form an $(18,6,5)$
pair; an avoider--container pair intersects in at most 5 points and is always
compatible. An invariant code is therefore precisely an $(18,6,10)$ code
together with an $(18,6,5)$ code, so its size is at most
$\A(18,6,10)+\A(18,6,5)=\A(18,6,8)+\A(18,6,5)\le 427+72=499$ (complementation;
upper bounds from the live table~\cite{brouwer-table}), far below 2970.
\end{remark}

Search saturations are reported with a measured calibration factor
(greedy-only probes underestimate plateau-search caps by ${\approx}1.25\times$,
obtained by running both on the known 23-cycle type --- a single-type
calibration, so corrected caps are indicative rather than certified); every
corrected cap remains below 2970, with the largest reported cap approximately
17\% below target. The multi-cycle 5- and 7-type margins are thinner than the
raw probes suggest, and the constraint-relaxation near-misses repair to well
below target across all types. Per-type orbit censuses, the exact counts behind the
aggregated rows, and all saturation logs are under \path{data/}
(Section~\ref{sec:data}).

\paragraph{Summary.}
With the single exception of cycle type $13+1^{10}$, every
prescribed-automorphism structure whose order is divisible by a prime $\ge 5$
is either excluded exactly by a theorem or analytic cap, or yields tightly
clustered search saturations below the target. The exact exclusions and
measured caps together suggest that a code of size
at least 2970 is unlikely to possess an automorphism of prime order $\ge 5$
outside the open $13$-type, consistent with the 1990 label ``no known
structure.'' (Groups
whose order is supported only on 2 and 3 are not swept; symmetry that weak
leaves an invariant space too large to bound with this method, and does not
constitute exploitable structure.) The $13+1^{10}$ type and the depth-4
exchange remain the two open items.

\section{The improved bounds: Johnson composition and the reconstructed 1990
improver}\label{sec:bounds}

\paragraph{Johnson decomposition.}
An $(n,d,w)$ code splits by its last coordinate into $A\times\{0\}$ (an
$(n-1,d,w)$ code) and $B\times\{1\}$ (an $(n-1,d,w-1)$ code), with a cross
constraint between $a\in A$ and $b\in B$ of $|a\cap b|\le w-d/2$: the new
coordinate adds 1 to each cross distance, and $d(a,b)=2w-1-2|a\cap b|$ is
odd, so $d\bigl((a,0),(b,1)\bigr)\ge d$ if and only if $|a\cap b|\le w-d/2$
--- one intersection unit weaker than the constraint
$|b\cap b'|\le(w-1)-d/2$ inside the weight-$(w-1)$ half. Fixing $A$ to the
best known $(n-1,d,w)$ code and searching $B$ in its full cross-compatible
pool (henceforth the \emph{cross-pool}) reduces each cell to a single
maximum-weight independent set problem on a dense conflict graph.

\paragraph{Search strength.}
With $A$ fixed to the surviving \texttt{a22.6.10.2183H} (2183 words), the
weight-9 cross-pool has 20{,}756 words; our own earlier searchers had
repeatedly capped its maximum compatible subset near 660, giving
total 2843, below the incumbent 2969. CHILS~\cite{chils}, a recent dense-graph MWIS heuristic, instead reaches 796 across a short ladder
($660 \to 768 \to 779 \to 784 \to 790 \to 796$) of independent seeded runs on
a laptop. The ``${\approx}660$ cap'' was an artifact of search strength, not
of structure.

\paragraph{The four improved bounds.}
All four codes pass both independent verifiers and the offline record gate of
Section~\ref{sec:verification}, and now appear on the live table.

\begin{itemize}
\item $\A(23,6,10)\ge \mathbf{2979}$ (lost 2970, surviving 2969).
  Construction: \texttt{a22.6.10.2183H} $\times\{0\}$ $\cup$ $B\times\{1\}$,
  where $B$ is a set of 796 mutually compatible weight-9 words found by CHILS
  in the 20{,}756-word cross-pool. $2183+796=2979$.
\item $\A(24,6,10)\ge \mathbf{4214}$ (lost 4200, surviving 4174).
  Construction: our $\A(23,6,10)$ code of 2978 words $\times\{0\}$ $\cup$
  $B\times\{1\}$, where $B$ is a set of 1236 weight-9 words found by CHILS in
  the corresponding 54{,}930-word cross-pool. $2978+1236=4214$. (This cascades
  the first result. The $A$-half is an earlier 2978-word code from the same
  pipeline --- not a subcode of the 2979-word code; the two share 2377 words
  --- recoverable from the released 4214-word file by restricting to words
  ending in 0; the composition was not re-run on the later 2979-word code.)
\item $\A(23,6,11)\ge \mathbf{3539}$ (improves the surviving 3535; the lost
  3585 is not reached). Construction: \texttt{a22.6.11.2636H} $\times\{0\}$
  $\cup$ $B\times\{1\}$, $B$ a set of 903 weight-10 words. $2636+903=3539$.
  Exact MWIS solves over unions of diverse CHILS solutions (\emph{solution
  merges} --- exact only on the merged subpool, hence \textsc{measured}
  evidence about the full cross-pool) saturate at 901--903 across rounds; the
  lost value would need $|B|=949$: 46 beyond our best.
\item $\A(24,6,8)\ge \mathbf{1855}$ (improves the surviving 1848; the lost
  1882 is not reached). Construction: \texttt{a23.6.8.1439H} $\times\{0\}$
  $\cup$ $B\times\{1\}$, $B$ a set of 416 weight-7 words. $1439+416=1855$. A
  solution merge over an earlier round (a union not containing the final
  416-word solution) has independence number 414 (\textsc{optimal} on that
  subpool); re-attaining the lost value would require $|B|=443$: 27 beyond
  our best.
\end{itemize}

The other five lost $d=6$ cells were not reached by this construction, with
saturation evidence (\textsc{measured}): $\A(24,6,12)$ caps at 5536 ($<$ the
surviving 5558) via a solution-merge optimum of 1997; $\A(25,6,8)$
and $\A(26,6,8)$ yield compositions below their surviving codes; and
$\A(27,6,8)$ and $\A(28,6,8)$ are out of reach of our searches a fortiori
(larger $B$-targets against tighter pools). The method's wins are exactly the
small-gap, pool-rich cells; its walls are the large-gap weight-8 chain and
the composition-capped weight-12 cell.

\paragraph{The reconstructed 1990 improver.}
The 1990 paper~\cite[\S XII]{bsss1990} describes the improver behind the lost
values: apply a random coordinate permutation; push codewords
lexicographically downward; delete the lexicographically last $k$ words ($k$
up to ${\approx}20\%$ of the code); refill. $k$-optimality was certified only
for $k=2$--$5$ --- so the lost values lived at the end of large-hole
trajectories, fully consistent with our depth-3 dryness theorem
(Theorem~\ref{thm:depth3}). We reimplement this improver faithfully, adding
one capability the 1990 authors lacked: at each plateau, the exact depth-$K$
exchange of Section~\ref{sec:exchange} runs on the current state, so a
plateau is either escaped by an exact witness or proved exchange-dry to
depth $K$ (no improving exchange deleting at most $K$ words) before
perturbation. On $\A(18,6,6)$, a 30-word lexicographic hole punched into the
known optimum 133 re-closes to exactly 133 in 32 cycles, chaining four exact
witnesses; on the $(23,6,10)$ cell itself, an interrupted run ratcheted
$2795\to 2807$ through chained witnesses.

\section{Reproduction distance and the difficulty ladder}\label{sec:distance}

How far was the lost 1990 result from what unaided greedy construction
reaches? We first fix the baseline. The \emph{lexicode} is the greedy code:
enumerate all weight-10 words in lexicographic order and accept each word
compatible with every word accepted so far.

\begin{proposition}\label{prop:lexicode}
The $(23,6,10)$ lexicode in the standard coordinate order has exactly 2377
words.
\end{proposition}

The size is empirically robust to relabeling: eight independent random
coordinate permutations each yield exactly 2377 (\textsc{measured};
lexicographic greedy is not relabeling-equivariant in general, so we do not
claim invariance). Hence the 1990 pipeline's annealing-and-refill stages
accounted for a climb of ${\approx}593$ words above the greedy baseline.

The full ladder is: lexicode $2377$ $\to$ surviving incumbent $2969$ $\to$
lost value $2970$ $\to$ the live table's upper bound
$7521$~\cite{brouwer-table}. Reproducing a lost value is therefore a
${\approx}600$-word climb, not a ${\approx}30$-word polish, and an improver
campaign needs iteration counts commensurate with the months of runtime the
1990 search spent on period hardware. Modern semidefinite and linear
programming upper-bound work
(Schrijver-type~\cite{schrijver2005}; Kim--Toan~\cite{kim-toan};
Polak~\cite{polak2019}; cf.\ also~\cite{avz2000}) never touched this cell and
leaves roughly $2\times$ lower/upper gaps in the nearest $d=6$ cells, so upper
bounds posed no obstacle to 2970 or to our 2979.

Surpassing 2970 and 4200 shows that at least two of the nine lost
values were attainable (indeed undershot). The remaining question, for the
seven not surpassed here, is whether they are attainable, or whether some are
1990 overcounts --- the $\A(18,6,6)$ $144\to133$ precedent keeps this
possibility live. Both possibilities motivate either large-scale reproduction
runs or reclassification, cell by cell.

\section{Data and code availability}\label{sec:data}

Every statement above is machine-checked or measured; all scripts, logs, and
code files reproduce it from a clean clone of the repository
\url{https://github.com/Chrislysen/solon}.

\paragraph{Explicit codes.}
{\sloppy
The four codes accompany this paper as ancillary files (Brouwer format:
one codeword per line), and are also hosted in the reference table's code
archive as \texttt{a23.6.10.2979H}, \texttt{a24.6.10.4214H},
\texttt{a23.6.11.3539H}, and \texttt{a24.6.8.1855H} under
\url{https://aeb.win.tue.nl/codes/cwc/d6/} (verified live, July 2026).\par}

\paragraph{Verifiers.}
\path{anc/independent_verify.py} (standalone, dependency-free,
re-checks all four codes; included as an ancillary file); the engine's dual
verifiers; the offline record gate \path{scripts/cwc_gate_dump.py}.

\paragraph{Pipelines.}
{\sloppy
\path{scripts/cwc_johnson_chain.py} (composition), the CHILS integration,
\path{scripts/cwc_hole_refill.py} (exchange computations),
\path{scripts/cwc_orbit_group.py} (symmetry sweep), and
\path{scripts/cwc_coadapt.py} (reconstructed improver). Logs are under
\path{data/}.\par}

\paragraph{Environment.}
All timings are from a single consumer laptop (Intel Core Ultra~9 275HX,
16\,GB RAM, Windows~11). OR-Tools CP-SAT~9.15 was used; the CBC, SCIP, and
CHILS version identifiers, exact command lines, and per-run logs are recorded
under \path{data/}.

\paragraph{Status.}
The four bounds were submitted to the table's maintainer on 20 July 2026,
independently verified by him, and added to the live
table~\cite{brouwer-table} in July 2026 with the corresponding code files;
the table's lost-codes list for $d=6$ has dropped from nine entries to seven.

\paragraph{Note added (21 July 2026).}
After verifying the submitted codes and updating the table, A.~E.~Brouwer
applied his local-improvement (``polish'') program to them and improved all
four: $\A(23,6,10)\ge 2981$, $\A(23,6,11)\ge 3550$, $\A(24,6,8)\ge 1857$, and
$\A(24,6,10)\ge 4216$. The live table now lists these polished values, with
code files (\texttt{a23.6.10.2981H} etc.)\ that we have independently
re-verified; our original four codes remain available as the ancillary files
of this paper and in the table's archive under their original names. The
analysis in this paper concerns our original codes.

\subsection*{Acknowledgments}
\phantomsection\addcontentsline{toc}{subsection}{Acknowledgments}
The author thanks Andries E.\ Brouwer for maintaining the reference table,
for promptly and independently verifying the submitted codes and updating the
table, and for the polish improvements recorded in the closing note. Part of this work was conducted while the author was an exchange
student at the University of California, Berkeley.
Computational searches were conducted using the author's Solon research
software. The reported lower bounds are independently verifiable from the
released code files, and the exact optimization claims are accompanied by
complete model descriptions and reproducibility logs.

\phantomsection

\end{document}